\documentclass[11pt,fleqn]{article}
\usepackage{fullpage}
\usepackage{amsmath,amssymb,amsthm}
\usepackage{flushend}

\newcommand{\HH}{\mathcal{H}}

\newcommand{\veps}{\varepsilon}

\newcommand{\E}{{\mathbb E}}
\newcommand{\etal}{{\it{et al.\ }}}

\DeclareMathOperator{\polylog}{polylog}
\DeclareMathOperator{\diam}{diam}

\DeclareMathOperator{\ddim}{ddim}
\DeclareMathOperator{\rad}{rad}
\DeclareMathOperator{\MST}{MST}
\DeclareMathOperator{\MStT}{MStT}
\DeclareMathOperator{\OPT}{OPT}
\DeclareMathOperator{\RS}{RS}
\DeclareMathOperator{\SBG}{SBG}
\DeclareMathOperator{\PD}{\Delta \Phi}

\newtheorem{theorem}{Theorem}[section]
\newtheorem{lemma}[theorem]{Lemma}

\theoremstyle{definition}
\newtheorem{definition}[theorem]{Definition}

\newcounter{this-list}

\newcounter{par-list}
\newlength{\parlistlength}

\begin{document}

\title{Near-linear time approximation schemes for \\
Steiner tree and forest in low-dimensional spaces}

\author{
Yair Bartal \\
Hebrew University\\
Email: yair@cs.huji.ac.il
\and
Lee-Ad Gottlieb \\
Ariel University\\
Email: leead@ariel.ac.il
}

\maketitle

\begin{abstract}
We give an algorithm that computes a $(1+\veps)$-approximate Steiner forest
in near-linear time 
$n \cdot 2^{(1/\veps)^{O(\ddim^2)} (\log \log n)^2}$.
This is a dramatic improvement upon
the best previous result due to Chan \etal \cite{Chan-16}, who
gave a runtime of about 
$n^{2^{O(\ddim)}} \cdot 2^{(\ddim/\veps)^{O(\ddim)} \sqrt{\log n}}$.

For Steiner tree our methods achieve an even better runtime
$n (\log n)^{(1/\veps)^{O(\ddim^2)}}$ in doubling spaces. For Euclidean space the runtime can be reduced to $2^{(1/\veps)^{O(d^2)}} n \log n$, 
improving upon the result of Arora \cite{{A-98}} in fixed dimension $d$.
\end{abstract}

\newpage

\section{Introduction}
In the Steiner tree and forest problems we are given a set $S$ of points, as well
as a subset $X \subset S$. The set $S-X$ constitutes the Steiner points, and
$X$ the set of real points. The Steiner tree problem is to find a miminal 
spanning tree for $X$, where the tree may also utilize points of $S-X$.
In the Steiner forest problem, we are also given a set $L$ of pairs or {\em terminals},
and the goal is to find a collection of trees on $S$ 
in which every terminal pair is found in the same tree, and for which the
sum of the edge-lengths in all the trees (the trees' weights) is minimized.
Clearly, $|X|-1$ terminals suffice to describe the desired point groupings.

The Steiner tree problem is NP-hard to approximate within a factor 
$\frac{96}{95}$ \cite{CC-08}, and admits a $1.39$-approximation \cite{BGRS-10}.
Arora \cite{A-98} considered the Euclidean case and presented a
PTAS\footnote{PTAS, which stands for a Polynomial-Time 
Approximation Scheme,
means that for every \emph{fixed} $\veps>0$ there is a $(1+\veps)$--approximation.
Note that for every constant $\veps>0$, the runtime is polynomial in $n$.
} 
for this problem with run time $n (\log n)^{\veps^{-O(d)}}$ in $d$ dimensions.
The Steiner forest problem inherits the same hardness of approximation of
$\frac{96}{95}$ from Steiner tree, 
and Agrawal \etal \cite{AKR-95} presented a 2-approximation for this problem. 
Borradaile \etal \cite{BKM-15} investigated the problem in the 
more restrictive setting of the Euclidean plane. Building upon techniques of 
Arora \cite{A-98} and Mitchell \cite{M-99} along with several innovations, 
they presented an algorithm that gives a $(1+\veps)$-approximate Steiner
forest in time $n \log^{O(1/\veps)}n$. 
Similarly, Bateni \etal \cite{BHM-11} considered the restriction of the problem to planar
graphs and graphs of bounded treewidth and obtained a $(1+\veps)$-approximation
in polynomial time. 

Borradaile \etal suggested that their construction may extend to 
low-dimensional Euclidean space, but if true this seems a non-trivial task.
More recently Chan \etal \cite{Chan-16} presented a PTAS
for low-dimensional Euclidean spaces, and in fact for the more general 
low-dimensional doubling spaces. To prove this result, they built off the
travelling salesman framework of \cite{BGK-16}, while utilizing several
additional (and quite sophisticated) techniques. Their randomized
algorithm yields a $(1+\veps)$-approximate Steiner forest with high
probability, in total time
$n^{2^{O(\ddim)}} \cdot 2^{(d/\veps)^{O(\ddim)} \sqrt{\log n}}$.

\paragraph{Our contribution.}
We revisit the Steiner forest problem, and prove the following:

\begin{theorem}\label{thm:main}
There exists a deterministic 
$n \cdot 2^{(1/\veps)^{O(\ddim^2)} (\log \log n)^2}$-time approximation scheme for the 
Steiner forest problem in doubling spaces.
\end{theorem}

We note that \cite{Chan-16} achieve a PTAS only for constant dimension, 
and even in this regime our runtime is only 
$n \cdot 2^{(1/\veps)^{O(1)} (\log \log n)^2}$,
a strict and significant improvement over
$n^{O(1)} \cdot 2^{(1/\veps)^{O(1)} \sqrt{\log n}}$.
Further, our result greatly increases the range of $\ddim$ for 
which a PTAS is possible from constant to 
$\ddim \le \sqrt{c \log_{1/\veps} \log n}$
(for some constant $c$).
The goal of achieving near-linear time approximation schemes
(similar to those of \cite{A-98, EKM-12, BG-13, BKM-15})
has long been a central focus of research, and we are first to do this
for Steiner forest. 
And finally, our algorithm is deterministic.

We are aided in our task by a decomposition theorem of Eisenstat \etal \cite{EKM-12}
(Theorem 2.2, see also \cite{BHM-11})\footnote{
This construction can be implemented in time  $\veps^{-O(\ddim)} n \log n$ on a doubling spanner.}
which states that it is 
sufficient to find an approximation to the Steiner problem whose cost
is that of the optimal solution plus $\veps$ times the minimum spanning tree of the space.
To this we add the following contributions:
\begin{itemize}
\item
Our first contribution is that there exists an approximate Steiner forest with a very
simple structure: It's composed of shallow Steiner binary trees wherein all
internal points are Steiner points, and the shallow trees are joined to each
other at the leaves, which are real points (Section \ref{sec:proper}).
This holds for all metric spaces (not only those of low dimension), and
as such constitutes a fundamental contribution to the study of Steiner 
trees and forests.
\item
Building on this construction, we show how to identify a small 
universal subset of the Steiner points sufficient for constructing an approximate 
Steiner forest.
These Steiner points also admit a light spanning graph (Section \ref{sec:subset}).
The indentification of these points allows us to build a light {\em forest banyan} for the 
space, that is a graph which weighs only a constant factor more than the minimum
spanning tree on the real points,
yet to which the Steiner forest may be restricted (Section \ref{sec:forest-banyan}).

Previously, the existence of forest banyans was known only for planar graphs \cite{EKM-12},
and it was not known whether doubling or even Euclidean spaces admit forest banyans.
We believe this result to be of independent interest.
\item
We then utilize a decomposition technique that allows us to consider only spaces
that have sparse Steiner forests
(Section \ref{sec:sparsity}).
A similar approach was used by \cite{BG-13}
for travelling salesman tours, and by \cite{Chan-16} for Steiner forest.)
Indeed, our main theorem, Theorem \ref{thm:metric-forest}, assumes 
sparse spaces, and gives an approximation guarantee that is additive 
in the weight of the $\MST$ of the space.
\item
Finally, we present a clustering technique and 
associated dynamic program which
allows the computation of a low-weight Steiner forest
(Section \ref{sec:dynamic}). 
While the dynamic program has elements of those of \cite{BKM-15,Chan-16},
we introduce several new innovations that allow the dynamic program
to exploit the sparsity of the underlying forest banyan.
In particular, our clustering is novel in the type of clusters it creates:
Each cluster has children of various radii, where 
the child sub-clusters near the boundary of the parent cluster all have small 
radius, while the child clusters farther in have progressively larger radii.
This construction allows us to achieve a far superior run time to what was
previously known.
\end{itemize}

Although our focus is on Steiner forest, our forest banyan 
(Theorem \ref{thm:banyan}) in conjunction with the simpler clustering of
\cite{BG-13} or the Euclidean approach of \cite{RS98}
immediate imply the following improvements for Steiner tree:

\begin{theorem}
The Steiner tree problem can be solved in time
$n (\log n)^{(1/\veps)^{O(\ddim^2)}}$
in doubling spaces, and in time
$2^{(1/\veps)^{O(d^2)}} n \log n$
in Euclidean spaces.
\end{theorem}

\section{Preliminaries and notation}

\paragraph{Graphs.}
The weight of an edge $e$ ($w(e)$) is its length. The weight of 
an edge-set $E$ is $\sum_{e \in E} w(e)$, and the weight of a graph
$G=(V,E)$ is the weight of its edge-set, $w(G)=w(E)$.

Let $B(u,r) \subset V$ refer to the vertices of $V$ contained in the 
closed ball centered at $u \in V$ with radius $r$.
$B^*(u,r)$ is the edge set of the complete graph on $B(u,r)$.

\begin{definition}
A graph $G=(V,E)$ is $q$-{\em sparse} if
for every radius $r$ and vertex $v \in V$, the weight of
$B^*(v,r) \cap E$ is at most $qr$. 
\end{definition}

\paragraph{Steiner trees.}
Consider a point set $S$ endowed with a metric distance function.
Let a subset $X \subset S$ be the set of `real' points, and $S-X$ 
the Steiner points. The minimum spanning tree of $X$ 
($\MST(X)$) does not use points of $S-X$, while the minimum Steiner
tree of $X$ on $S$ ($\MStT(X) = \MStT_S(X)$) may. An edge is called a 
{\em real edge} if it connects
two points of $X$, and otherwise it is called a Steiner edge.
It is well-known that $w(\MStT(X)) \le w(\MST(X)) \le 2 w(\MStT(X))$, 
and a corollary of this is that for sets $A \subset B$ we have
$w(\MST(A)) \le 2w(\MST(B))$.

Any Steiner tree can be modified to be binary:
After choosing a root, Steiner nodes with only one child
can be bypassed and deleted, while nodes with more children
may have this number reduced by duplicating the node, assigning
the duplicate as a child of the node at distance 0, and dividing
the original children among the node and its duplicate.

\paragraph{Doubling dimension and hierarchies.}
For a point set $S$, let $\lambda = \lambda(S)$
be the smallest number such that every
ball in $S$ can be covered by $\lambda$ balls of half the
radius, where all balls are centered at points of $S$.
Then $\lambda$ is the {\em doubling constant} of $S$,
and the {\em doubling dimension} of $S$ is
$\ddim = \ddim(S)=\log_2\lambda$ \cite{As-83}.
The dimension is often taken to be an integer by rounding
up the real number. The following is the well-known 
packing property of doubling spaces (see for example \cite{KL-04}):
If $S$ is a metric space and $C \subseteq S$ has
minimum inter-point distance $b$, then
$|C| = \left( \frac{2\rad(S)}{b} \right)^{O(\ddim)}$.

Similar to what was described in \cite{GGN-06,KL-04}, a point-set $X$
is a $\gamma$-net of $Y$ if it satisfies the following properties:
\renewcommand{\labelenumi}{(\roman{enumi})}
\begin{enumerate}
\item Packing: For every $x,y \in X$, $d(x,y) \ge \gamma$.
\item Covering: Every point $y \in Y$ is strictly within distance $s$ of
some point $x \in X$: $d(x,y) < \gamma$.
\end{enumerate}
The previous conditions require that the points of $X$ be spaced out, yet
nevertheless cover all points of $Y$. A point in $X$ covering a point in
$Y$ is called a {\em parent} of the covered point; this definition allows
for a point to have multiple parents.

A hierarchy $\HH$ for set $S$ is composed of a series of nets, 
where each level of the hierarchy is both a subset and  net of the level beneath it.
We shall assume throughout (and without loss of generality) that the
minimum inter-point distance is 1.
For $i=0,\ldots,P$ (where $P := \lceil \log \diam(S) \rceil$),
fix $H_i\subseteq S$ to be an $2^i$-net of $S$,
called the net of \emph{level} $i$, or of \emph{scale} $2^i$.
Notice that the bottom hierarchical level $H_0$ contains all points,
and the top level $H_P$ contains only a single point.
Throughout this paper, we will assume without loss of generality
(as in \cite{A-98} and
subsequent papers) that the aspect ratio of the space is $n^{O(1)}$,
and so a hierarchy has $O(\log n)$ levels. We can build this
hierarchy in time $2^{O(\ddim)}n \log n$ \cite{KL-04,HM-06,CG-06}.

\paragraph{Spanning trees and spanners.}
The following lemma, due to Talwar \cite{T-04} (see also \cite{S-10}),
uses the doubling dimension to bound the weight of the minimum spanning
tree of any metric graph. It is an adaptation of a similar statement 
of Arora \cite{A-98} for Euclidean spaces.

\begin{lemma}\label{lem:mst}
For point set $S' \subset S$ we have
$w(\MST(S')) \le 4 |S'|^{1-{1}/{\ddim(S)}} \cdot\diam(S')$.
\end{lemma}

Let $G=(V,E)$ be a metric graph, where vertices $V$ represent points of
some metric set $S$, while the edge weights of $E$ correspond to inter-point
distances in $S$.
A graph $R = (V_R,E_R)$ is a $(1+\veps)$-stretch {\em spanner} of $G$
if $R$ is a subgraph of $G$
(specifically, $V_R = V$ and $E_R \subset E$), and also
$d_R(u,v) \le (1+\veps) d_G(u,v)$
for all $u,v \in V$.
Here, $d_G(u,v)$ and $d_R(u,v)$ denote the shortest path
distance between $u$ and $v$ in the graphs $G$ and $R$, respectively.

It is known that 
Euclidean spaces admit $(1+\veps)$-stretch spanners with {\em lightness}
$W_E = \veps^{-O(d)}$, meaning that the total spanner weight is at most
a factor $W_E$ times the weight of the $\MST$ of the set \cite{DHN93}.
Recently, a series of papers \cite {G-15, FS-16, BLW-19} considered light 
spanners in doubling spaces, ultimately demonstrating that they admit
spanners of lightness $W_D = \veps^{-O(\ddim)}$ \cite{BLW-19}.
In both the Euclidean and metric case, the well-known greedy spanner 
achieves this weight bound;
this is the spanner that considers all edges in order of weight, and adds
the current edge to the spanner if the current stretch on the partial spanner
is greater than $1+\veps$. A variant of this can be computed in time 
$\veps^{-O(\ddim)} n \log n$, and possesses only $\veps^{-O(\ddim)}$ edges.



\section{A forest banyan}\label{sec:forest-banyan}

In this section we show how to identify a set of Steiner points
$S' \subset S-X$ that are sufficient for use in constructing an 
approximate Steiner forest for $X$. Our forest banyan will simply be 
a spanner constructed on top of $X \cup S'$. 
To accomplish the identification, we will require a deep result about Steiner trees:
In Section \ref{sec:proper} we will show that every Steiner tree can be represented as 
a collection of shallow Steiner trees connected at the leaves by real edges
(edges that connect two points of $X$).
The existence of a shallow Steiner tree decomposition will then be
used to identify the set $S'$ (Section \ref{sec:subset}):
If the shallow Steiner subtree can be closely approximated by a tree
on $X$ alone, we can simply connect these points with spanner edges for $X$.
But if such a construction is impossible, it implies for these points a certain
excess of weight in respect to their diameter. We will use this property to
identify a small (and light) set of Steiner points that can be added to
$S'$.

\subsection{Proper Steiner tree decompositions}\label{sec:proper}

Given point-set $S$ and real subset $X \subset S$,
define a {\em proper} Steiner tree for $X$ to be 
a binary tree in which the leaves are real points
(points of $X$) and the internal nodes are Steiner points
(points of $S-X$).
If $P$ is a proper Steiner tree, then 
$X(P)$ denotes the real points of $P$, that is its leaves.
We can represent any Steiner tree $T$ as a collection of 
proper Steiner trees $\cal{P}$ connected at the leaves by a set $E$ of 
real edges, so we denote $T = ({\cal P},E)$, with any pair
$P,P' \subset \cal{P}$ connected by at most one real edge in $E$.
We call $({\cal P},E)$ the {\em proper decomposition} of $T$.\footnote{
Note that ${\cal P}$ is a set of proper Steiner trees, and these include Steiner edges.
So the notation $({\cal P},E)$ should not be conflated with the standard
graph notation $(V,E)$, where the first set contains only vertices.
}

Suppose some edge $e_{v,w}$ in $P$ connects a parent $v$ to its child
$w$; then we define $P_w$ 
to be the proper Steiner tree rooted at $w$, and $P_v$ 
to be the proper Steiner tree formed by pruning
$P_w$ from $P$, and connecting $v$'s remaining child in 
$P$ directly to $v$'s parent, thereby bypassing and deleting $v$.
(That is, if $u$ is the parent of $v$, and $v$ has two children
$w,x$, then $v$ is removed, and $x$ becomes a child of $u$.
In the special case where $v$ is the root, then $v$ is simply removed, 
and its two children $w,x$ become the roots of their respective subtrees.)
$P_v,P_w$ are not connected, and so splitting $P$ results in the 
splitting of $T$ into trees $T_v,T_w$ ($P_v \subset T_v, P_w \subset T_w$).
We prove the following:

\begin{lemma}\label{lem:decomp}
Given point-set $S$, subset $X \subset S$ and any $0<\veps<1$, 
there exists a Steiner tree $T = ({\cal P},E)$ 
satisfying
\begin{enumerate}
\item
Low weight: $w(T) \le (1+\veps)w(\MStT(X))$ 
\item
Separation:
For all $P \in {\cal P}$ and edge $e_{v,w} \in P_i$,
the replacement of $P$ by $P_v,P_w$
would result in the splitting of $T$ into
Steiner trees $T_v,T_w$ 
($P_v \subset T_v, P_w \subset T_w$)
for which
$$d(X(T_v),X(T_w)) \ge w(e_{v,w}) + \veps w(P_w).$$
\end{enumerate}
\end{lemma}

\begin{proof} 
The procedure to produce $T$ is as follows:
We initialize $T=({\cal P},E)$ to be the (proper decomposition of the) 
minimum Steiner tree of $X$ on $S$.
For any $P \in \cal{P}$ and parent-child pair $v,w \in P$,
say that $P_w$ {\em violates} the separation property if the removal of edge
$e_{v,w} \in P$ and replacement of $P$ by $P_v,P_w$ splits $T$ into trees 
$T_v,T_w$
for which $d(X(T_v),X(T_w)) < w(e_{v,w}) + \veps w(P_w)$. 
We locate an edge $e_{v,w} \in P$ for which $P_w$ violates the 
separation property and for which $w(P_w)$ is the minimum over all violators.
We then remove $P$ and replace it with $P_v,P_w$, and reattach 
$T_v,T_w$ by adding to $E$ a single real edge connecting the closest points between 
$X(T_v)$ and $X(T_w)$.
Now $T=({\cal P},E)$ is the newly produced Steiner tree, 
and the procedure is repeated upon it until no violators remain.

Turning to the analysis, the second item follows by construction. 
For the first item, the removal of an edge $e_{v,w} \in P$
decomposes $P$ into two proper Steiner trees $P_v,P_w$.
A new edge is then added to $E$ to reconnect $T_v,T_w$, 
and we will charge the weight of this new edge to 
$e_{v,w}$ and $P_w$ (where $P_w$ will be charged only an $\veps$-fraction
of its own weight):
Edge $e_{v,w}$ (upon being deleted from $E$) is charged exactly its own weight $w(e_{v,w})$,
while the rest of the charge for the new edge is divided among the edges of $P_w$
proportional to their lengths, so that each edge in $P_w$ is charged less than $\veps$ times
its length. We will show that the edges of $P_w$ cannot
be charged again at a later stage of the procedure, from which we conclude that 
the sum of all charges to subtrees is at most $\veps$ times the original weight of $T$.
Since the weight of edges deleted from $T$ were replaced with identical weight,
we conclude that the final tree weighs less than $(1+\veps)$ times the original
weight of $T$, and the first item follows.

It remains to show that each edge is charged only once.
We will show that once $P_w$ is added to the collection it is never split again.
Consider an edge $e_{a,b}$ in $P_w$: 
Since $e_{a,b}$ was not removed before $e_{v,w}$, and since 
$w(P_b) < w(P_w)$, it must be that $P_b$ was not a violator.
We show that once $P_b$ is not a violator, it cannot subsequently become a violator:
Consider $X(T_a), X(T_b)$. 
Suppose some split breaks off a subset of $X(T_a)$; for example
$X_1(T_a)$ is broken off of $X(T_a)$, leaving subset $X_2(T_a)$.
Recall that by construction, subset $X_1(T_a)$ is reattached to the closest point in
$X_2(T_a) \cup X(T_b)$.
If $X_1(T_a)$ is reattached to a point of $X_2(T_a)$, then set $X(T_a)$
remains the same, as does the distance between $X(T_a), X(T_b)$.
Otherwise, $X_1(T_a)$ is reattached to a point of $X(T_b)$, 
and so by construction it must be that
$d(X_1(T_a), X_2(T_a)) 
\ge d(X_1(T_a), X(T_b))
\ge d(X(T_a), X(T_b))$.
Then the distance between the new sets is
$d(X_2(T_a), X_1(T_a) \cup X(T_b)) 
= \min \{d(X_2(T_a), X_1(T_a)), d(X_2(T_a), X(T_b)) \} 
\ge d(X(T_a), X(T_b))
$.
The analysis for a split of $X(T_b)$ is the same.
\end{proof}

We now strengthen the result of Lemma \ref{lem:decomp} to give a decomposition
into {\em shallow} proper Steiner trees.

\begin{theorem}\label{thm:decomp}
Given point-set $S$, subset $X \subset S$ and any $0<\veps<1$, 
there exists a Steiner tree $T = ({\cal P},E)$ 
satisfying
\begin{enumerate}
\item
Low weight: $w(T) \le (1+4\veps) w(\MStT(X))$ 
\item
Separation:
For all $P \in {\cal P}$ and edge $e_{w,v} \in P$,
the replacement of $P$ by $P_v,P_w$
would result in the splitting of $T$ into
Steiner trees $T_v,T_w$ 
($P_v \subset T_v, P_w \subset T_w$)
for which
$$d(X(T_v),X(T_w)) \ge w(e_{v,w}) + \veps w(P_w).$$
\item
Depth and cardinality: Each $P \in \cal{P}$ has depth at most 
$t = (8/\veps) \ln (2/\veps)$,
and so it has at most $2^t$ leaves and $2^t-1$ internal nodes.
\end{enumerate}
\end{theorem}

\begin{proof}
We begin with the tree $T$ implied by Lemma \ref{lem:decomp}, and 
restrict ourselves to pruning subtrees from the proper Steiner subtrees
in $T$. Then the separation property is preserved.

Consider in turn each proper Steiner subtree $P \in {\cal P}$,
and impose on the points of $P$ an arbitrary left-right sibling ordering,
which implies an ordering on all points. 
Let $r$ be the root of $P$, 
and $R \in S-X$ the set of internal nodes of $P$
at hop-distance exactly $t$ from $r$. 
Let ${\cal R}$ be the set of proper
Steiner trees rooted at the points of $R$, with tree 
$R_i \in {\cal R}$ rooted at Steiner point $r_i \in P$.

We will first show that $\sum_i w(R_i) \le \veps w(P)$.
Consider any subtree $P_x$ of $P$ rooted at some point $x$.
If $L_x$ is the path that visits the leaves of subtree $P_x$ in order, then 
we can show that
$\frac{w(P_x)}{2} \le w(L_x) \le 2w(P_x)$:
The upper-bound on $w(L_x)$ follows from the fact that $L_x$ is not longer than an
in-order traversal of all the nodes of $P_x$. 
The lower-bound on $w(L_x)$ follows by induction:
The base case is when the subtree is a leaf, and for the inductive step let
$y,z$ be the children of $x$, with respective subtrees and in-order
paths $P_y,L_y,P_z,L_z$. By the inductive assumption we have
$w(L_y) \ge \frac{w(P_y)}{2}$
and
$w(L_z) \ge \frac{w(P_z)}{2}$,
and recalling the separation guarantee we have
\begin{eqnarray*}
w(L_x) 
&\ge& w(L_y) + d(X(P_y),X(P_z)) + w(L_z)	\\
&\ge& w(L_y) + \max \{ d(x,y), d(x,z) \} + w(L_z)	\\
&\ge& \frac{w(P_y)}{2} + \frac{d(x,y)+d(x,z)}{2} + \frac{w(P_z)}{2}	\\
&=& \frac{w(P_x)}{2}.
\end{eqnarray*}
We then conclude that
$\frac{w(P_x)}{2} \le w(L_x) \le 2w(P_x)$.
Using this fact, and again applying the separation property, 
we can further show that 
\begin{eqnarray*}
w(L_x)
&\ge& w(L_y) + d(X(P_y),X(P_z))	+ w(L_z)	\\
&\ge& w(L_y) + \veps \max \{ w(P_y),w(P_z) \} + w(L_z)	\\
&\ge& w(L_y) + \frac{\veps}{2} (w(P_y)+w(P_z)) + w(L_z)	\\
&\ge& (1+\frac{\veps}{4})(w(L_y) + w(L_z)).
\end{eqnarray*}
Crucially, this means that as the tree is traversed upwards,
the weight of the tree edges grows at a factor of at least
$(1+\frac{\veps}{4})$
per level, which implies that the weights of the paths
grow exponentially as the tree is traversed upwards.
Defining in-order paths $L,L_i$ respectively for $P,R_i$,
and noting that when $0 \le x \le 1$ we have $(1+x) \ge e^{x/2}$,
we conclude that
$$w(P)
\ge \frac{w(L)}{2}
\ge (1+\frac{\veps}{4})^t \sum_i w(L_i)
\ge e^{t\veps/8} \sum_i w(L_i)
\ge \frac{2}{\veps} \sum_i w(L_i)	
\ge \frac{1}{\veps} \sum_i w(R_i)	
$$
as claimed.

Now modify $P$ as follows:
For each subtree $R_i$ defined above, remove all its edges from $P$, 
and add to $E$ the $\MST$ of $X(R_i)$.
Also remove the edges connecting $r_i$ (the root of subtree $R_i$) to its parent,
and instead connect the parent to the closest point in $X(R_i)$.
This completes the construction for each $P \in {\cal P}$,
and the resulting graph clearly satisfies the
separation and depth properties of the lemma. 
For the weight bound: The weight of each subtree $P$ does not increase,
and the $\MST$ edges added to $E$ in place of the subtrees $R_i$
weigh at most $\sum_i 2w(R_i) \le 2\veps w(P)$. 
Incorporating the $(1+\veps)$ weight term of Lemma \ref{lem:decomp},
we have total weight $(1+\veps)(1+2\veps)<(1+4\veps)$.
\end{proof}

\subsection{Identifying Steiner subsets}\label{sec:subset}

In the previous section we showed
that every Steiner tree can be replaced by a set of
proper Steiner trees connected at their leaves, where each tree has
$2^t = 2^{\tilde{O}(1/\veps)}$
leaves (Theorem \ref{thm:decomp}). 
In this section, we will
use this property to identify a subset of the Steiner points
sufficient for approximating the proper Steiner trees.
Let $t = \tilde{O}(1/\veps)$ be as defined in Theorem \ref{thm:decomp}.
We prove the following:

\begin{theorem}\label{thm:banyan}
Given sets $X \subset S$ and a parameter $\veps$, we can in time
$2^{O(t \ddim)} |X| \log |S|$ identify a subset $S' \subset S$
for which 
$\MStT_{S'}(X) \le (1+\veps)\MStT_{S}(X)$,
and further, 
$w(\MST(X \cup S')) = \veps^{-O(\ddim)} w(\MST(X))$.
\end{theorem}

Then the forest banyan is simply a light $(1+\veps)$-stretch spanner
on the points of $X \cup S'$.
The rest of this section details the proof of Theorem \ref{thm:banyan}. 
We begin by showing how to identify a subset of the Steiner points
sufficient for approximating the proper Steiner trees for a single
cluster (a subset of $X$), where we pay some cost in the diameter
of the cluster.

\begin{lemma}\label{lem:ball}
Given a subset 
$D \subset X \subset S$,	
point-set $S$ equipped with a hierarchy, and
parameter $0<\veps<1$;
in time $\veps^{-O(\ddim)} |D| \log |S|$ we can identify a 
subset $S' \subset S$ with the following properties:
\begin{enumerate}
\item
Size: $|S'| = \veps^{-O(\ddim)} |D|$.
\item
Weight: $w(\MST(D \cup S')) = \veps^{-O(\ddim)} w(\MST(D))$.
\item
Approximation:
Every subset $V \subset D$ with $|V| \le 2^t$ 
that admits a shallow proper Steiner tree $P$ on $S$, 
also admits a Steiner tree on $S'$ of weight at most 
$(1+\veps)w(P) + \veps \diam(D)$.
\end{enumerate}
\end{lemma}

\begin{proof}
The construction is as follows:
Set 
$r = \diam(D)$
and 
$r' = \frac{\veps r}{2^t}$.
For each scale 
$i=\lfloor \log r' \rfloor, \ldots, \lceil \log (2rt) \rceil$, 
we consider all points of $S$ within distance $2^i$ of $D$, and
then extract from $S$ a $2^i(\veps^2/64)$-net $S_i$ for these close
points. Set the final solution to be $S' = \cup_i S_i$.

The stated runtime is easily achieved by using the hierarchy
to identify net-points in $S$ for the points of $S$ 
sufficiently close to $D$.
To prove the bound on size:
First note that index $i$ takes only 
$O(t + \log (1/\veps)) = O(t)$ 
distinct values.
For each $S_i$, consider a set $D' \subset D$ 
that is a $2^i$-net of $D$.
By construction, all points of $S_i$
are within distance $2 \cdot 2^i$ of $D'$. Then 
$|S_i| = \veps^{-O(\ddim)} |D'| 
= \veps^{-O(\ddim)} |D|$,
and so 
$|S'| \le \sum_i |S_i| 
= t \veps^{-O(\ddim)} |D| 
= \veps^{-O(\ddim)} |D|$,
as claimed.

To prove the bound on weight:
For any $p \in D'$ we have that
$w(\MST(B(p,2\cdot 2^i) \cap S_i)) = \veps^{-O(\ddim)}2^i$.
Since 
$S_i = \cup_{p \in D'} [B(p,2\cdot 2^i) \cap S_i]$
we have that
\begin{eqnarray*}
w(\MST(D \cup S_i)) 
&=&	w(\MST(D \cup_{p \in D'} [B(p,2\cdot 2^i) \cap S_i]))		\\
&\le& 	w(\MST(D)) + \sum_{p \in D'} w(\MST(B(p,2\cdot 2^i) \cap S_i))	\\
&=& 	w(\MST(D)) + \veps^{-O(\ddim)}2^i |D'| 				\\
&=& 	\veps^{-O(\ddim)} (w(\MST(D))+2^i) 				\\
&=& 	\veps^{-O(\ddim)} w(\MST(D)),
\end{eqnarray*}
where the penultimate equality follows from the fact that $D'$ is a $2^i$-net,
implying that
$2w(\MST(D)) \ge w(\MST(D')) \ge 2^i (|D'|-1)$.
Now we can form $\MST(D \cup S')$ by taking 
$\cup_i \MST (B \cup S_i)$, 
which increases the total weight by a factor of $O(t)$. 
We conclude that 
$w(\MST(D \cup S')) = \veps^{-O(\ddim)} w(\MST(D)$,
as claimed.

To prove the approximation bound we use a charging argument.
Consider a Steiner point $p \in S$ in the 
shallow proper Steiner tree $P$ of $V \subset D$ on $S$. 
Let $a,b$ be the childen of $p$ with respective subtrees
$P_a,P_b \subset P$, and define $L,L_a,L_b$ to be the respective
leaf traversals for $P,P_a,P_b$.
As in the proof of Theorem \ref{thm:decomp} we have 
$\frac{w(P)}{2} \le w(L) \le 2w(P)$, 
and also that
$w(L) \ge (1+\frac{\veps}{4})(w(L_a)+w(L_b))$;
the latter inequality implies that
the increase in $w(L)$ over $w(L_a),w(L_b)$ is
$w(L)-w(L_a)-w(L_b) > (\veps/8) w(L)$.

Let $f$ be the distance from $p$ to its closest leaf descendant,
and trivially $f<w(P)$. 
Since the distance between any two points in $D$ is at most $2r$,
the separation property of Theorem \ref{thm:decomp} implies that
every edge in $P$ is not larger than the maximum distance between
points of $X(P)$, that is $2r$. 
Noting that the depth of $P$ is at most $t$, we have $f \le 2rt$.
Now take the value of $i$ which satisfies $2^{i-1} < f \le 2^i$,
and clearly $i \le \lceil \log (2rt) \rceil$.
If $i<\lfloor \log r' \rfloor$, then the cost of moving $p$ to
its closest net-point in $S_i$ is less than $\frac{\veps r}{2^t}$,
and the cost of moving $|V|$ such points is less than
$|V|\frac{\veps r}{2^t} \le \veps r = \veps \diam(D)$.
Otherwise $\lfloor \log r' \rfloor \le i \le \lceil \log (2rt) \rceil$,
and so $S_i$ contains some Steiner point within distance 
$2^i(\veps^2/64) 
< 2f(\veps^2/64) 
< 2w(P)(\veps^2/64)
\le 4w(L)(\veps^2/64)
= w(L)(\veps^2/16)$
of $p$. 
The replacement of $p$ by this Steiner point is charged to the
weight increase of $L$ over $L_a,L_b$ --
the increase (shown above to be $(\veps/8) w(L)$) is charged a 
$\frac{w(L)(\veps^2/16)}{(\veps/8) w(L)} = \frac{\veps}{2}$ 
fraction of its weight --
from which it follows that the total cost of replacing all 
such points of $P$ by nearby Steiner points is 
$\frac{\veps}{2} w(L) \le \veps w(P)$.
\end{proof}

Lemma \ref{lem:ball} would be useful if for each (unknown) 
shallow proper Steiner tree, we could identify
a ball of similar radius which contains it, and to which we
can charge the cost of the new Steiner points. 
This identification is indeed possible, as in the following lemma.

\begin{lemma}\label{lem:balls}
Given $X \subset S$, a hierarchy for $X$, and 
and parameter $0<\veps<1$;
in time $\veps^{-O(\ddim)} |X| \log |X|$ we can identify a 
set $\mathcal{D}$ containing clusters
$D \in \mathcal{D}$ ($D \subset X$) satisfying:
\begin{enumerate}
\item
Cardinality: $|{\mathcal D}| = \veps^{-O(\ddim)}|X|$,
and $|D| = 2^{O(t \ddim)}$
for each $D \in {\mathcal D}$.
\item
Covering:
Every subset $V \in X$ of size $|V|\le 2^t$
either has a light Steiner tree on $S$ of weight 
$w(\MStT(V)) \le (1+\veps)\diam(V)$, 
or else some set
$D \in {\mathcal D}$ 
of diameter $O(\diam(V))$
is a 
$\frac{\veps \diam(V)}{|V|}$-net of $V$.
\item
Low weight:
$\sum_{D \in {\mathcal D}} w(\MST(D)) = \veps^{-O(\ddim)}w(\MST(X))$.
\end{enumerate}
\end{lemma}

\begin{proof}
For our analysis, we define a potential function on collections
of balls and inter-ball edges.
The potential $\Phi$ of a collection is the sum of diameters of 
balls in the collection, plus the portion of each edge 
not contained in a ball. That is, 
(i) if edge $e_{a,b}$ is not contained in any
ball, it contributes $w(e_{a,b})$ to $\Phi$;
(ii) if $a$ is contained in a ball or multiple balls but 
$b$ is not, then let $B(x,r)$ be the smallest ball
containing $a$, and the edge contributes 
$w(e_{a,b}) - (r-d(x,a))$ to $\Phi$;
and
(iii) if both $a$ and $b$ are contained in balls, then let
$B(x,r), B(x',r')$ be the smallest balls
containing $a$ and $b$ respectively, and the edge contributes
$\max \{0,w(e_{a,b}) - (r-d(x,a)) - (r'-d(x',b))\}$
to the potential.


Our procedure to build collections is as follows:
First define collections $C_i \in {\mathcal C}$
for $i=0,\ldots,t+ \lceil \log (80/\veps) \rceil-1$.
Collection $C_i$ has radius parameter $r_i = 2^i$,
while the $j$-th level 
($j=0,\ldots,\frac{\lceil \log \diam(X) \rceil}{t}$)
of collection $C_i$ has radius parameter 
$r_{i,j} = 2^{j[t+\lceil \log (80/\veps) \rceil]}r_i 
= 2^{j[t+\lceil\log (80/\veps) \rceil]+i}$.
This means that each possible ball radius is represented 
(up to a factor of 2) in some level of some $C_i$.
Now initialize each $C_i$ to contain the edges of the 
light $(1+\veps/4)$-spanner edges of $X$, and no balls. 
Then the initial potential of $C_i$ is
$\Phi_i \le W_D \cdot w(\MST(X))$.
Our analysis will consider each $C_i$ separately.

We now detail an {\em insertion step}
that adds some ball $B$ to collection $C_i$,
while removing from $C_i$ all edges and smaller balls 
fully contained within $B$.
Define the {\em potential decrease} of $B$ 
with respect to $C_i$ ($\PD_i(B)$) to be
the difference in potential of $C_i$ before an insertion
step of $B$, and the potential of $C_i$ after such an insertion step. 
Now our procedure for collection $C_i$ considers $r_{i,j}$ for increasing $j$:
Consider all $\veps r_{i,j}$-net points of $X$,
and investigate the balls centered at these points with radii in the 
range $[r_{i,j}, 4r_{i,j}]$;
these balls are considered in order of maximum potential decrease.
For a candidate ball $B$, we undertake an insertion step for
$B=B(x,r)$ if all the following conditions are met:
\begin{enumerate}
\item
$\PD(B) \ge \frac{\diam(B)}{10}$.
\item
$B$ fully contains at least one edge, 
or one ball of $C_i$ not centered at $x$.
\item
$B$ is not within distance $r_{i,j}/\veps$ of some ball of 
the same parameter $r_{i,j}$ previously added to $C_i$.
\end{enumerate}
The procedure on $C_i$ terminates after considering all relevant
balls for all $j$.

The construction of ${\mathcal D}$ is as follows:
For every ball $B(x,r)$ added to some collection $C_i$
in the above procedure, define a set $D$ to be the 
$\delta$-net of $B(x,2r/\veps) \cap X$,
for 
$\delta=\frac{\veps r}{8 \cdot 2^t}$.
Add every $D$ to ${\mathcal D}$, and this completes our
construction.

The cardinality of ${\mathcal D}$ follows 
from the second insertion step condition:
The spanner has at most $\veps^{-O(\ddim)}|X|$
edges, so $\veps^{-O(\ddim)}|X|$ balls can be
added to $C_i$ in place of edges. Likewise,
only $X$ balls can be added to $C_i$ in place
of smaller balls. So
$|{\mathcal D}| = t\veps^{-O(\ddim)}|X| = \veps^{-O(\ddim)}|X|$.  
The size of each $D \in {\mathcal D}$ is
immediate from the construction.

For the claim of covering:
Consider set $V$ with minimum Steiner tree
on $X$ of weight at least 
$(1+\veps) \diam(V)$. 
By construction, some $i,j$ satisfies that 
$\frac{\diam(V)}{2} \le r_{i,j} < \diam(V)$.
We will show below that some ball with this radius parameter
contains $V$ and also has potential decrease at least
$\frac{\veps}{4} \diam(V)$.
Then either this ball is added 
to $C_i$, or a nearby ball has already been added --
in either case, by construction the corresponding set $D$
contains a 
$\delta = \frac{\veps r}{8 \cdot 2^t}
\le \frac{\veps 4r_{i,j}}{8 \cdot 2^t}
\le \frac{\veps \diam(V)}{2^t}
$-net for $V$, as required.

Now take 
the minimal weight Steiner tree $T$
for $V$ on $X$, where $T$ is restricted to edges of
the light $(1+\veps/4)$-stretch spanner on $X$. 
Consider the state of these edges at the beginning of 
the procedure for $C_i$ and $j$:
If all edges of $T$ were present in the collection,
they would contribute at least 
$w(\MStT(V))$ to $\Phi_i$.
Now some edges of $T$ may have been removed by insertion steps, 
and in their place sit balls of some parameter $r_{i,j-1}$ or less,
but we can show that the remaining edges of $T$ along
with the added balls still contribute at least
$(1-\veps/4)w(\MStT(V))
\ge (1+\veps)(1-\veps/4)\diam(V)
\ge (1+\veps/2)\diam(V)$
to $\Phi_i$:
Take a ball $B=B(x,r)$ sitting on edge of $T$, 
and we analyze the loss of potential due to $B$.
To do this, we will charge against paths of $T$, 
where we decompose $T$ into paths at points of $V$ and
at points of degree three or more in $T$.
Clearly, this decomposes $T$ into fewer than $2|V|$ paths,
and so the average length of a path is greater than
$\frac{w(T)}{2|V|}$.
We consider two cases:
\begin{enumerate}
\item
$B$ covers or cuts some edge of $T$, and the path $P$
on this edge continues in both direction a distance of
more than $(1+\veps/4)r$. In this case we can show that $B$ 
does not cover or cut an edge of any other path:
Otherwise we could modify $T$ by deleting from $P$ a
segment of length greater than $(1+\veps/4)r$ beginning 
inside $B$, and then repair $T$ by connecting the newly disjointed
path to the second path in $B$ using spanner edges,
at a cost of at most $(1+\veps/4)r$. This would reduce 
the weight of $T$, contradicting the optimality of $T$.

Then let $a,b$ be the first and last points of the single path
within the ball. By the stretch property of the spanner and
definition of the potential function,
the deleted segment accounted for a loss in potential of at most
$(1+\veps/4)d(a,b) + (r-d(x,a)) + (r-d(x,b))
\le (1+\veps/4)\diam(B)$, 
while $B$ itself contributes $\diam(B)$ to the potential.
Then $B$ accounts for a loss of at most $(\veps/4)\diam(B)$
to the potential of $C_i$, and all balls of this type together 
account for a loss of at most 
$\frac{\veps w(T)}{4}$ to $\Phi_i$.
\item
$B$ covers or cuts some edge of $T$, and the path $P$
on this edge continues in some direction a distance of
less than 
$(1+\veps/4)r \le 4(1+\veps/4)r_{i,j-1} \le 5r_{i,j}$. 
In this case the ball may contain multiple paths, so we
cannot charge to the ball -- we shall charge the path
instead.
Each path can be charged the above loss for each end,
amounting to at most
$2|V| \cdot 10 r_{i,j-1}
< \frac{\veps}{8 \cdot 2^t} r_{i,j}
\le \frac{\veps}{8 \cdot 2^t} \diam(V)
\le \frac{\veps w(T)}{8|V|}$.
The sum of losses in potential for all edges (fewer than
$2|V|$ edges) is less than $\frac{\veps w(T)}{4}$.
\end{enumerate}
We can now show that some ball of parameter $r_{i,j}$
results in a sufficient decrease in $\Phi_i$:
Take the point $p \in V$ for which a spanner path exits $V$
to a distance of at least $\diam(V)$, and let 
$p'$ the closest $\frac{\veps r_{i,j}}{4}$ net-point to
$p$. Let $q \subset V$ be the farthest point from $p'$, and set 
$r = (1+\veps/4)d(p',q) + 4r_{i,j-1}$,
and then $B(p',r)$ includes the edges of $T$ and all balls of 
parameter $r_{i,j-1}$ covering those edges.
It also includes a second path rooted at $p$ exiting the
ball, and as above that path (or rather, the remaining
edges of that path along with balls replacing those
edges), contributes
$(1-\veps/4)(r - \frac{\veps r_{i,j}}{4})$ 
to the potential decrease of $B(p',r)$. Then 
$\Phi_i(B(p',r)) 
\ge (1+\veps/2)\diam(V) + (1-\veps/4)(r- \frac{\veps r_{i,j}}{4}) -2r
\ge (2+\veps/4)(r-\frac{\veps r_{i,j}}{4}) - 2r
\ge \frac{\veps \diam(B)}{10}$, which
completes the proof of covering.

For the weight bound: 
If for some removed ball $B=B(x,r)$ we have that
$w(\MST(D)) = \veps^{-O(\ddim)}\diam(B)$,
then the bound is immediate by construction, 
as the potential decrease is at least $\veps \diam(B)$.
Assume then that $w(\MST(D))$ is greater than this.
In this case, it suffices to show that  
$\PD(B) = \Omega(w(\MST(D')))$, where $D' = D \cap B$: 
By the maximality of $\PD(B)$ and the packing property, this 
result would imply that 
$\PD(B) = \veps^{-\Theta(\ddim)} \Omega(w(\MST(D)))$,
from which the weight claim follows.
To prove the lower-bound on the potential decrease of $B$,
we use the same approach as in the proof of the covering property above:
Let $T$ be the minimum spanning tree of $D'$ restricted to
the light spanner edges of $X$. At the time when $B$
is chosen to be added to collection $C_i$,
some edges of $T$ may have been previously removed,
and in their place sit balls of parameter 
$r_{i,j-1}$ or less. As above, all these balls can be 
charged to the fewer than $|D'|$ paths of $T$.
Since $D'$ is a net, the average length of a path is at least
$\frac{w(\MST(D'))}{|D'|}
\ge \frac{|D'-1|\veps r/(4 \cdot 2^t)}{|D'|}
> \veps r/(8 \cdot 2^t)$.
To each path we charge the loss of its ends, which
is at most $\veps r/(16 \cdot 2^t)$.
So the potential decrease of $B$ is at least 
$w(\MST(D'))/2$,
which completes the proof.
\end{proof}

Theorem \ref{thm:banyan} follows by first noting that 
Theorem \ref{thm:decomp} implies a solution consisting
of only proper Steiner trees connected at the leaves --
and that each tree has at most $2^t$ leaves, that is
real points.
Applying Lemma \ref{lem:balls}
to $X,S$, we identify a set ${\mathcal D}$ of clusters
that form nets for all trees, and a Steiner tree
on an $\frac{\veps \diam(V)}{|V|}$ net of $V$ implies 
a $(1+\veps)$-approximate tree on $V$.
Applying Lemma \ref{lem:ball} on each of these clusters
to identify the subset $S' \subset S$, we achieve
the bounds of Theorem \ref{thm:banyan}.

\section{Inducing sparsity}\label{sec:sparsity}

In the previous section, we demonstrated that one can identify a 
subset $S' \subset S$ to which an approximate Steiner forest of 
$X$ may be retricted, and further that 
$w(\MST(X \cup S')) = \veps^{-O(d)}w(\MST(X))$. 
Then a spanner on $X \cup S'$ is a forest banyan for $X$ on $S'$
of weight $W_B \cdot w(\MST(X))$ where $W_B=\veps^{-O(\ddim)}$

In this section, we show that given a graph $G$ that constitutes
a forest banyan for $X$, we can replace $G$ by a set of sparse
graphs -- that is graphs in which every possible ball 
contains only a relatively light set of edges --
solve Steiner forest separately on each sparse graph,
and join the separate solutions into a single solution for
the original graph. The replacement by sparse graphs
will be necessary for running the dynamic program
of Section \ref{sec:dynamic}. 
Several results on TSP and Steiner forest in doubling spaces 
use a procedure to segment the space into sparse areas
\cite{BGK-16,BG-13,G-15,Chan-16,CJ-16},
and the following theorem is related to those results.

To prove the theorem below, we will utilize net-respecting graphs, 
introduced in \cite{BGK-16}:
A graph $G$ is said to be $\delta$-\emph{net respecting (NR)} relative to a given hierarchy
$\{H_i\}_{i=0}^{\lceil \log \diam(S) \rceil}$ and
value $0<\delta<\frac{1}{4}$, if for every edge in $G$, say of length $\ell$, both of its endpoints
belong to $H_i$ for $i$ such that $2^i \leq \delta \ell < 2^{i+1}$.
As demonstrated in \cite{BGK-16}, 
given any graph we can reroute it to be net-respecting as follows:
Consider each edge in increasing order. If edge $e_{x,y}$ is not net-respecting, remove
it and add the edge between the closest $i$-level net-points to $x$ and $y$, respectively.
Also connect $x,y$ to their respective closest $i$-level net-points via net-respecting paths.
It is easy to show that all these operations add a stretch factor of at most $1+O(\delta)$.

\begin{theorem}\label{thm:sparse-forest}
Given point-sets $X,S'$ ($m=|X|+|S'|$), 
a hierarchy for these points,
a set $L$ of terminal pairs for
$X$, and a connected graph $G=(V,E)$ on $X \cup S'$ 
of weight $w(G) = W_B \cdot w(\MST(X))$
to which a Steiner forest for $X$  is restricted; 
we can in time $O(m \log m)$ identify subsets 
$X_0,\ldots,X_p$
satisfying
$\cup_i X_i = X$,
new terminal assignments 
$L_0,\ldots,L_p$
and corresponding connected net-respecting graphs 
$G_0,\ldots,G_p \in {\mathcal G}$
(where $L_i$ contains only points in $X_i$ 
and $X_i \subset V_{G_i}$)
with the following properties:
\begin{enumerate}
\item
Sparsity: Each $G_i$ is $q = W_B^{O(\ddim)}$-sparse.
\item
Forest cost: 
Let $\OPT$ be the optimal Steiner forest for 
$X$ restricted to $G$, and 
$\OPT_i$ be the optimal Steiner forest for
$X_i$ restricted to $G_i$.
Then 
$\cup \OPT_i$ is a Steiner forest for $X$ on $S'$
(but not restricted to $G$), while
$\sum_i w(\OPT_i) 
\le \OPT + \frac{\veps}{W_B} w(G)
\le \OPT + \veps w(\MST(X))$.
\item
Total size: 
$|V| = \veps^{-O(\ddim)}m$.
\end{enumerate}
\end{theorem}

\begin{proof}
For some $\delta = \Theta(\veps)$,
reroute the edges of $G$ to create a $\delta$-net-respecting 
graph $G'$ of weight at most $(1+\frac{\veps}{2})w(G)$.
Fix $q=(W_B)^{c\ddim}$ for a sufficiently large constant $c$ to
be specified later.
Beginning with the lowest level of the hierarchy
and iterating upwards, consider each $i$-level net-point:
let $p$ be the $i$-level net-point for which
$w(B^*(p,2 \cdot 2^i) \cap E_G)$ 
is maximum: If this weight is greater than $q 2^i$,
we have encountered a heavy ball that must 
be removed from $G$ in order to induce sparsity.

The removal of the heavy $i$-level ball proceeds as follows:
First, we will assume inductively that all $j$-level balls
$(j<i)$ are already $O(q)$-sparse,
and then we can show that the weight
of any ball $B^*(p,r) \cap E_G$ for $r=O(2^i)$ is at most
$2^{O(\ddim)}qr$:
This ball covers $2^{O(\ddim)}$ $q$-sparse balls of level $i-1$,
the contribution of edges smaller than $2^i$
to the weight of $B^*(p,r) \cap E_G$ is at most $2^{O(\ddim)}qr$.
Since $G$ is $O(\veps)$-net-respecting, the contribution
of edges of length $[2^i,2r]$ to the weight of $B^*(p,r) \cap E_G$ 
is at most $\veps^{-O(\ddim)}r$.
Having established a bound on the weight of $B^*(p,r) \cap E_G$,
an averaging argument implies that there must be specific radius
$r \in [2 \cdot 2^i,4 \cdot 2^i]$
for which $B(p,r)$ cuts at most 
$2^{O(\ddim)}q$ edges of $G$ of length $2^i$ or less
(where a cut edge has one vertex inside
the ball and the other outside of it).
Having determined the ball $B(p,r)$, 
we remove from $V_G$ all points in $B(p,r) \cap V_G$, 
except for the $2^{O(\ddim)}q$ points incident 
upon short edges of length at most $2^i$ exiting $B(p,r)$. 
We also retain in $G$ the net-points upon which 
may be incident the net-respecting edges longer than $2^i$. 
These amount to 
$f = 2^{O(\ddim)}q + \veps^{-O(\ddim)} = 2^{O(\ddim)}q$
total points.
All edges with removed endpoints are removed from $E_G$,
and to $E_G$ is added a net-respecting minimum spanning
tree for all remaining points of $B(p,r)$. 

Consider this ball $B(p,r)$ be the the $k$-th removed ball.
Create a new graph $G_k$, and place in it copies of all 
points and edges of $B(p,r)$, as well as a copy of the same 
minimum spanning tree that was added to $E_G$.
We note that by constrution $G_k$ is $2^{O(\ddim)}q$-sparse
(as the added minimum spanning tree can account for only 
$\veps^{-O(\ddim)}$ new edges in any ball).  
Create a terminal list $L_k$ for $G_k$, and
modify the terminal list $L$ of $G$: 
For any pair with both points removed from $G$,
remove that pair from $L$ and place it in $L_k$.
For any pair with exactly one point removed from $G$,
retain the pair in $L$ but with the removed point replaced
by $p$, and add to $L_k$ a copy of the pair with the retained
point replaced by $p$. 
Finally, add $G_k$ to ${\mathcal G}$, fix $X_k$ to be the real 
points in $G_k$, set $X=X-X_k$, and continue the procedure
on $G$, searching for the next heavy ball to be removed.

Clearly, upon removal of $G_k$ from $G$,
the two solutions for the Steiner problem on 
$G,L$ and on $G_k,L_k$ can be joined 
into a valid solution for the original problem, 
with the only additional forest weight being the two added spanning trees. 
Recall that the minimum possible weight of the removed ball is $qr$,
and by Lemma \ref{lem:mst}
we have that the weight of the two added spanning trees is at most
$2(1+\veps) \cdot 4 f^{1-1/\ddim(S)} \cdot 2r
= 2^{O(\ddim)} q^{1-1/\ddim(S)} r
= \frac{\veps qr}{W_B}$
(for a sufficiently large constant $c$ in the definition of $q$).
The additional weight is charged to the removed edges,
and the forest cost condition is satisfied.

The size follows by construction, since the removal of
a ball necesarily removes points from $F$, while the removed
ball may duplicate at most $\veps^{-O(\ddim)}$ points.
\end{proof}

\section{Clustering and dynamic programming}\label{sec:dynamic}

We can now give the dynamic program that given $X,S$ and the terminal 
pairs $L$ of $X$, computes an approximate Steiner forest for $X$ on $S$. 
We assume that we can compute a forest banyan for the space (Theorem \ref{thm:banyan}), 
and further that the forest banyan is net-respecting and $q$-sparse (Theorem \ref{thm:sparse-forest}).
We first describe a nested hierarchical clustering of the space, that is a clustering where a cluster in level 
$i$ completely contains the points of its ($i-1$)-level cluster children.
We prove that there exists
a Steiner forest with favorable properties with respect to the clustering, and then
show that there exists an efficient dynamic program for Steiner forest and the given
clustering. 

Before describing the clustering, we need some preliminary 
definitions and proofs:
Define annulus $A(x,r_1,r_2) = B(x,r_2) - B(x,r_1)$.
For parameters integer $k$, 
$\delta = \frac{r}{2^k}$ 
and $0 \le i \le k$,
define
$A_i(r,\delta) = A(x, r - \delta 2^i, r - \delta \lfloor 2^{i-1} \rfloor )$
and 
$\bar{A}_i(r,\delta) = A(x, r + \delta \lfloor 2^{i-1} \rfloor, r + \delta 2^i)$.
That is, $A_0(r,\delta)$ is the outer annulus of $B(x,r)$ of width $\delta$, and
$A_1(r,\delta)$ is internal to $A_0(r,\delta)$ with width $2 \delta$. 
$\bar{A}_0(r,\delta)$ is the annulus of width $\delta$ immediately external 
to the ball, and $\bar{A}_1(r,\delta)$ is external to 
$\bar{A}_0(r,\delta)$ and of width $2\delta$, etc.
Let $\alpha_i(r,\delta,b)$ and $\bar{\alpha}_i(r,\delta,b)$ 
be the smallest number of balls of radius $b$
that can cover the points of $A_i(r,\delta) \cap P$ and $\bar{A}_i(r,\delta) \cap P$,
respectively.
We can prove the following lemma:

\begin{lemma}\label{lem:admit}
Let $G$ be a $q$-sparse net-respecting connected graph on points $P$.
Then given any center point $x$, parameters $r > \gamma > 0$ 
and value $c \ge 1$,
at least half of the radii $r' \in [r,2r]$ satisfy all the following 
conditions:
\begin{enumerate}
\item
Cut edges:
$B(x,r')$ cuts $O(q)$ edges of $G$ of length $r'$ or less.
\item
Sparse border:
$\alpha_0(r',\gamma,\gamma) = 2^{O(\ddim)}q$
and
$\alpha_0(r',c\gamma,\gamma) = 2^{O(\ddim)}qc$.
\item
Admissibility:
$\sum_{i=0}^{\lfloor \log (r'/\gamma) \rfloor}
\alpha_i(r',c\gamma,\gamma)
+ \bar{\alpha}_i(r',c\gamma,\gamma)
\le 2^{O(\ddim)}qc \log (r'/\gamma)$.
\end{enumerate}
\end{lemma}

\begin{proof}
For the first item, by sparsity the ball $B(x,4r)$ has weight
$O(qr)$, and so it has $O(qr)$ edges of length in the range $[1,2r]$,
meaning that any constant fraction (less than 1) of radii in the stated range
cut at most $O(q)$ edges.

For the second item and third items:
We can show that all points of $B(x,4r)$ can be covered by at most 
$2^{O(\ddim)}(qr/\ell) + \veps^{-O(\ddim)} = 2^{O(\ddim)}(qr/\ell)$ 
balls of radius $\ell < r$.
To see this, consider the components of $B(x,6r) \cap G$:
Since $G$ is $q$-sparse, the components with weight greater than
$\ell$ can all be covered greedily by at most 
$2^{O(\ddim)}(qr/\ell)$ balls of radius $\ell$.
For the smaller components, we may restrict our attention to
those intersecting $B(x,4r)$, and for these to be connected 
in $G$, it must be that they are incident upon an edge of 
length at least $r$ exiting $B(x,6r)$.
Since $G$ is net-respecting, edges of length 
$6r-4r-\ell \ge r$ or greater
are incident on $\veps r$-net points, of which there are
$\veps^{-O(\ddim)}$ within $B(x,6r)$; it follows that there are at most $\veps^{-O(\ddim)}$ 
components of this type, and each can be covered by a single
ball. 

A random choice of $r'$ induces random annuli. 
It follows that the points in a random annulus of width 
$c\gamma 2^i$ within $B(x,4r)$ can be covered by 
$2^{O(\ddim)}\frac{qr}{\gamma 2^i} \cdot \frac{c \gamma 2^i}{r} + \veps^{-O(\ddim)} = 2^{O(\ddim)}qc$ 
balls of radius $\ell = \gamma 2^i$ in expectation,
and the second item follows from Markov's inequality.
The third item follows by linearity of expection and Markov's inequality.
\end{proof}

\subsection{Clustering}

Our clustering is novel in the type of clusters it creates:
Each cluster will have at most 
$\veps^{-O(\ddim)} q^2 \log \log n$
child clusters with various radii. The child clusters near 
the boundary of the cluster will all have small radius, while
the child clusters farther in will have progressively larger radii 
(and hence Lemma \ref{lem:admit} above can be used to bound
the total number of child clusters necessary to cover the parent).

The clustering is constructed as follows:
Given sets $X \subset S$ and $S'$ along with a hierarchy
for these points, $0 < \veps<1$,
and a $q$-sparse banyan $G$ of weight $W_B \cdot w(\MST(X))$
(where $W_B = \veps^{-O(\ddim)}$ and 
$q = (W_B)^{O(\ddim)} = \veps^{-O(\ddim^2)}$),
fix $s$ to be the smallest power of 2 greater than 
$\left(\frac{2^{O(\ddim)}q W_B}{\veps} \cdot \frac{\log n}{\log \log n} \right)$.
We first build a hierarchy of {\em primary} levels bottom-up:
For $j= 0,1,\ldots$ create a $s^j$-net cluster 
(equivalently, a $2^i$-net cluster $i = j \log s$) thus:
For each $s^j$-net point $p$, choose a radius 
$r \in [2^i,2 \cdot 2^i]$ satisfying the guarantees of Lemma \ref{lem:admit}
for parameters $c = \veps^{-O(\ddim)}q$
and $\gamma = s^{j-1}$.
Define the $s^j$-net {\em cluster} centered at $p$ to be the union of all
$s^{j-1}$-net clusters (not previously claimed by a different $i$ level cluster) 
whose center points are found in $B(p,r)$.
We will call $B(p,r)$ a primary ball, and it creates a primary cluster.
(Note that a cluster is also determined by the $2^{O(\ddim)}$ balls 
of the same level which cut into its area before the cluster's ball 
was chosen.)
This gives a hierarchical clustering of 
$L = O(\log_s n)
=  ((\log n)/ \log s)
= O((\log n)/ \log \log n) 
= \frac{\veps}{2^{O(\ddim)}q W_B}s$
primary levels.\footnote{
As our algorithm will run in time exponential in $q$,
we have taken $q = O(\log n)$.}

For any $j$, let levels $i$ for $(j-1)\log s < i < j \log s$ be
the secondary levels. Secondary levels contain secondary clusters,
which are built recursively in a top-down fashion.
A primary or secondary cluster in level $(j-1)\log s + 1 < i \le j \log s$
is partitioned into $\veps^{-O(\ddim)} q^2 \log \log n$
child clusters as follows: For the $i$-level cluster, consider decreasing index 
$k = i-1,\ldots,(j-1)\log s + 1$. 
As above, identify in turn $k$-level balls with respective radii in 
the range $[2^k, 2\cdot 2^k]$ 
which satisfy the guarantees of Lemma \ref{lem:admit}
with respect to parameters $c$ as above and $\gamma = s^{j-1}$.
Each $k$-level ball must also satisfy that its center
$p$ is at distance at least $c2^k$ from the {\em boundary} of its parent cluster, 
meaning from the closest point not contained in the parent.
The new $k$-level cluster takes all points of the parent parent
that are within the $k$-level cluster (and have not yet been claimed
by a different $k$-level cluster).
After identifying all $k$-level balls, the procedure continues to 
identify $(k-1)$-balls -- in effect, the algorithm first adds large
child clusters near the center of the parent, then adds progressively
smaller clusters as it approaches the boundary. 
The new secondary clusters are then partitioned recursively, and so it follows 
from Lemma \ref{lem:admit}(iii) (with the above value of $c$) that any cluster has
$\veps^{-O(\ddim)} qc \log s = \veps^{-O(\ddim)} q^2 \log \log n$
children.

We now describe how portals are assigned to each cluster:
For an $i$-level cluster formed by an $r$-radius ball for 
$r \in [2^i, 2 \cdot 2^i]$, 
let its portals include those points incident upon edges of length 
at most $r$ cut by the forming ball, and also all
$\veps2^i$-net-points in the cluster
(upon which may be incident net-respecting edges of length 
greater than $2 \cdot 2^i$).
The cut edges include those that are directly cut by the ball,
and by Lemma \ref{lem:admit}(i) there are $O(q)$ such edges. Since a cluster is
formed by at most $2^O(\ddim)$ $i$-level balls there may be at most
$2^{O(\ddim)} q$ such edges cutting the cluster.
such edges. It also includes the edges cut by smaller primary
balls taken by the larger ball (if the other endpoint is not
in the larger ball). In this case however,
if the $i$-level cluster is also a primary cluster, 
it adds only the center of the smaller primary ball as 
a portal, and not the cut endpoint.
Together, these account for 
$2^{O(\ddim)} q + \veps^{-O(\ddim)} = 2^{O(\ddim)} q$ 
total portals.


We can show that there exists a forest with favorable properties with respect to the above clustering. 
The primary consistency property below is similar to that of Borradaile \etal \cite{BKM-15}, 
while the secondary consistency property is related to one found in Chan \etal \cite{Chan-16}.

\begin{lemma}\label{lem:clustering}
Given sets $X,S' \subset S$, $0 < \veps<1$, a list of terminal pairs, 
and a $q$-sparse net-respecting forest banyan $G$ of weight $W_B \cdot w(\MST(X))$ 
to which the Steiner forest is restricted, 
construct the above clustering with the given portals.
Then there exists a Steiner forest $F$ (not restricted to $G$) satisfying:

\begin{enumerate}
\item
Validity:
$F$ enters and exits between sibling clusters, 
and only via the respective portals of these clusters.
\item
Primary consistency:
Let $A$ be a $j\log s$-level primary cluster ancestral to some
$(j-1)\log s$-level primary cluster $B$. Then all paths of $F$ 
incident on points of $B$ and exiting ancestral cluster $A$
must be connected inside $B$, so that they belong to the same tree
in $F$.
\item
Secondary consistency:
Let (primary or secondary) cluster $A$ be the parent of secondary 
cluster $B$. If $B$ contains a terminal whose pair is outside $A$,
then all paths of $F$ incident on points of $B$ and exiting parent
$A$ must be connected in some ancestral cluster 
(but not necessarily inside $A$ or $B$), 
so that they belong to the same tree in $F$.
\item
Weight: 
$w(F) = \OPT + \frac{\veps w(G)}{W_B} \le \OPT + \veps w(\MST(X))$.
\end{enumerate}
\end{lemma}

\begin{proof}
For the first item, we show that the optimal forest $F_G$ 
restricted to graph $G$ can be slightly modified to satisfy this 
property. For every edge in $F$, if a cluster containing the
endpoint of the edge has that endpoint as a portal, then
$F$ is valid with respect to that edge. Most clusters which 
cut an edge have the endpoint as a portal - whether that 
edge was long with respect to the cluster radius (in which case
it is incident on a net-point portal), or whether the edge was
short, in which case the cluster added that endpoint as a 
portal. The exception is in the case where a $j\log s$-level 
primary cluster cuts a $(j-1)$-level primal cluster,
and then adds only its center as a portal.
In this case, we can reroute the edge through this center at a cost of
$O(2^{(j-1)\log s}) = O(s^{j-1})$.
We will charge this cost to the intersection of $G$ 
and the larger primary ball. This intersection is also present
in $2^{O(\ddim)}$ overlapping balls in each of $L$ levels,
so it can be charged at most 
$\frac{\veps s^j}{2^{O(\ddim)}L} \ge s^{j-1}$, as desired
(for appropriate choice of constants in the definition of $s$).

Secondary consistency: 
To achieve this property, we take $F_G$
along with the additional edges added for the first item. 
Call this forest $F_G'$. We note that after the addition of these 
edges the forest is still $O(q)$-sparse, and further that
the rerouted paths can be net-respecting, so 
$F_G'$ remains net-respecting.
We will add edges to $F_G'$ to create graph $F$ satisfying this property,
and then we will differentiate between original trees in $F_G'$ and
new trees incrementally formed in $F$: If a tree $T_G$ in $F_G'$ 
is within distance
$\frac{\min\{\diam(T_g), \diam(T) \} }{c} 
= \frac{\min\{\diam(T_g), \diam(T) \}}{\veps^{-\Theta(\ddim)} q}$ 
of a tree $T$ in $F$, 
then locate the smallest cluster
containing points of both $T_G,T$. If the diameter of the cluster is less than
$\frac{\min\{\diam(T_g), \diam(T) \} }{c}$ then
connect $T$ and $T_G$ at the portals throug which they exit this cluster.

It is immediate that the forest remains valid, and clearly $F$ satisfies the 
secondary consistency property. It remains only to bound the weight of added edges, 
and we will charge the new edges against $T_G$. 
We assume by induction that all new edges charged to a tree of $F_G'$
together weigh at most $\veps$ times the weight of the tree.
Consider tree $T$ and edge $e$ connecting $T_G$ and $T$.
If $T \cap F_G'$ has a long edge of length $\veps \diam(T_G)$
within distance $2 \diam(T_G)$ of $T_G$, then we charge the 
edge directly to $T_G$. There can only be $\veps^{-O(\ddim)}$
trees of $F$ with this property, so $T_G$ is charged at most
$\frac{\veps^{-O(\ddim)}w(T_G)}{c} < \veps w(T_G)$.
Assume then that $T \cap F_G'$ has no long edges within
distance $2 \diam(T_G)$ of $T_G$.
By the induction assumption, every new edge in $T$ is
incident on a tree of $F_G'$ of much greater weight,
and so by construction the weight of edges of
$T \cap F_G'$ within distance $2 \diam(T_G)$ of $T_G$
is at least $c w(e)$. But $F_G'$ is $O(q)$-sparse,
the weight of all edges of $F_G'$ within 
distance $2 \diam(T_G)$ of $T_G$ is 
$O(q w(T_G))$, and so the sum of the weights of all edges
charged to $T_G$ is at most 
$\frac{O(q w(T_G))}{c} < \veps w(T_G)$.

The proof of primary consistency is as follows:
Since a $s^j$-net cluster has only $2^{O(\ddim)}q$ portals, 
only $2^{O(\ddim)}q$ distinct trees exit the cluster.
Then the total cost of connecting trees that pass through common 
$s^{j-1}$-net child clusters at the exit portals, is at most
$2^{O(\ddim)}q s^{j-1} \le \frac{\veps s^j}{2^{\Theta(\ddim)}L}$,
as desired.
\end{proof}

\subsection{Dynamic program}

Given the clustering of the previous section, we present a dynamic program to
solve the Steiner forest problem. Specifically, we prove the following:

\begin{theorem}\label{thm:metric-forest}
Given the input and guarantees of Lemma \ref{lem:clustering},
we can find the optimal forest restricted to the validity
condition in time
$n \cdot 2^{q^{O(1)} (\log \log n)^2} 
= n \cdot 2^{\veps^{-O(\ddim^2)} (\log \log n)^2}$.
\end{theorem}

For metric spaces, Theorem \ref{thm:main} follows from 
Theorem \ref{thm:metric-forest} in conjunction with 
Theorem \ref{thm:sparse-forest} (sparsity),
along with Theorem 2.2 of \cite{EKM-12} (approximation to $\MST$).
The rest of this section presents the proof of
Theorem \ref{thm:metric-forest}.

\paragraph{Cluster description.}
As usual, the dynamic program will run bottom-up, computing the
optimal solution for each sub-cluster {\em configuration}, and using
these to compute optimal solutions for each parent configuration.
Each cluster, its children and portals have already been fixed.
A configuration for a primary cluster $C$ consists of:

\begin{enumerate}
\item
For each portal of $C$, a list of child clusters of $C$ with 
a path exiting that portal.
\item
For every portal pair of $C$, a boolean value indicating whether they
are connected in $C$.
\item
For every portal pair of $C$ not connected in $C$, a boolean value
indicating whether they must be connected outside of $C$.
\end{enumerate}

Since a cluster has $2^{O(\ddim)}q$ portals and 
$\veps^{-O(\ddim)} q^2 \log \log n$ child clusters,
a primary cluster may have $2^{q^{O(1)} \log \log n}$ 
different configurations. A secondary cluster $C$ has
another item in its configuration: For every primary
child $B$ of $C$, for every portal of $B$ a list of 
child clusters of $B$ with a path exiting $B$ via
that portal.
Since $C$ has at most 
$\veps^{-O(\ddim)} q^2$ primary child clusters
(Lemma \ref{lem:admit}(ii)), the number of 
configurations for a secondary cluster is similar
to that of the primary clusters.

\paragraph{Program execution.}
In executing the program, we assume the existence of a forest 
satisfying the guarantees of Lemma \ref{lem:clustering}, although
we will not necessarily enforce that the forest always 
obey all the requirements.

The computation of configurations for a cluster $C$ is as follows.
We try each combination of configurations for the child clusters of $C$,
and for each such combination we try all possible edge combinations 
connecting portals of $C$ and of its children. 
As the cluster has $q^{O(1)}\log \log n$ children, 
and each child have $2^{q^{O(1)}\log \log n}$ 
possible configurations, this amounts to $2^{q^{O(1)}(\log \log n)^2}$ work
per cluster.

We then check the
combination for {\em validity}, meaning we reject it if it does not obey 
the following conditions.
(Note that for the pair reachability condition below, 
we have not yet described how to to test whether 
a terminal pair is connected; this will be addressed below.)
\begin{enumerate}
\item
Primary consistency:
For a primary parent with a primary child, if multiple portals
of the child are incident upon disjoint paths exiting the parent, 
we require that these portals were marked in the child configuration
as internally connected.
(Note that in the program execution we assert primary consistency 
only when both father and child are primary clusters.)
\item
Secondary consistency:
For a child portal pair marked as needing to be connected
outside the child, these portals must either be connected in $C$,
or each have paths exiting $C$.
\item
Single terminal reachability:
For a child cluster $B$ with a terminal whose pair is outside 
$C$, there must be a path from a portal of $C$ to some portal of $B$ that
is marked as connected to the child cluster of $B$ containing the
terminal.
\item
Terminal pair reachability:
A terminal pair in child clusters $A,B$ must either be connected 
inside $C$, or else there must be paths from portals of $C$ to
portals of $A,B$ that are marked as connected to the grandchild
clusters containing the respective terminals.
\end{enumerate}
Of course, we do not permit the root cluster to have portals marked
as needing to be connected outside the cluster.

Having concluded that a combination is valid, we compute 
the configuration of $C$ implied by the combination:
For each portal of $C$, we compute the list of child clusters of $C$ with 
a path exiting that portal.
For each portal pair of $C$, we compute whether they are connected in $C$.
For a secondary cluster, we also record for each primary child cluster
the grandchild clusters connected to each of the child cluster portals. 
We must also mark some portals of $C$ as needing to be connected
ouside $C$. This occurs when a pair of portals of a child of $C$ are marked 
as needed to be connected outside the child, but are not connected
inside of $C$ and instead have disjoint paths reaching a pair of portals
of $C$. This also occurs in the pair reachability case mentioned above, 
when a portal pair is not connected in $C$ but reaches portals of
$C$ via disjoint paths; then these portals of $C$ must be marked as
connected.

It remains to explain how we verify whether a terminal pair residing
in child clusters $A,B$ is connected in $C$. We first note that the consistency and reachability conditions 
inductively imply that portals of $A$ or $B$ 
that are marked 
as being connected to the grandchild cluster containing the terminals,
are either connected to the terminals via internal paths, or marked
as needing to be connected to a different child portal that is connected 
to the terminal via an internal path.
Given $A,B$, we check whether the path connecting them is incident
on portals that are connected to the grandchild clusters containing 
the terminals. If $A,B$ are secondary clusters, secondary consistency
implies connectivity. If one or both are primary clusters, then we 
must also check that the path reaches the child of the primary cluster
containing the terminal, and then connectivity follows from 
either primary or secondary consistency (depending on the primary
cluster's child).

For each terminal pair, the program confirms that the pair is either
connected in the lowest cluster that contains them both, or else
that each terminal has a path exiting the cluster, and the exit
portals are marked as connected. Correctness follows.

\bibliographystyle{alpha}
\bibliography{sf}

\begin{thebibliography}{BGRS10}

\bibitem[AKR95]{AKR-95}
Ajit Agrawal, Philip Klein, and R.~Ravi.
\newblock When trees collide: An approximation algorithm for the generalized
  {S}teiner problem on networks.
\newblock {\em SIAM Journal on Computing}, 24(3):440--456, 1995.

\bibitem[Aro98]{A-98}
Sanjeev Arora.
\newblock Polynomial time approximation schemes for {E}uclidean traveling
  salesman and other geometric problems.
\newblock {\em J. ACM}, 45:753--782, 1998.

\bibitem[Ass83]{As-83}
P.~Assouad.
\newblock Plongements lipschitziens dans {${\bf R}\sp{n}$}.
\newblock {\em Bull. Soc. Math. France}, 111(4):429--448, 1983.

\bibitem[BG13]{BG-13}
Yair Bartal and Lee-Ad Gottlieb.
\newblock A linear time approximation scheme for {E}uclidean {TSP}.
\newblock In {\em Proceedings of the 2013 IEEE 54th Annual Symposium on
  Foundations of Computer Science}, FOCS '13, pages 698--706, 2013.

\bibitem[BGK16]{BGK-16}
Yair Bartal, Lee-Ad Gottlieb, and Robert Krauthgamer.
\newblock The traveling salesman problem: low-dimensionality implies a
  polynomial time approximation scheme.
\newblock {\em SIAM Journal on Computing}, 45(4):1563--1581, 2016.

\bibitem[BGRS10]{BGRS-10}
Jaroslaw Byrka, Fabrizio Grandoni, Thomas Rothvo\ss, and Laura Sanit\`{a}.
\newblock An improved lp-based approximation for steiner tree.
\newblock In {\em Proceedings of the Forty-second ACM Symposium on Theory of
  Computing}, STOC '10, pages 583--592, 2010.

\bibitem[BHM11]{BHM-11}
Mohammadhossein Bateni, Mohammadtaghi Hajiaghayi, and D\'{a}niel Marx.
\newblock Approximation schemes for {S}teiner forest on planar graphs and
  graphs of bounded treewidth.
\newblock {\em J. ACM}, 58(5):21:1--21:37, 2011.

\bibitem[BKM15]{BKM-15}
Glencora Borradaile, Philip~N. Klein, and Claire Mathieu.
\newblock A polynomial-time approximation scheme for {E}uclidean {S}teiner
  forest.
\newblock {\em ACM Trans. Algorithms}, 11(3), 2015.

\bibitem[BLW19]{BLW-19}
Glencora Borradaile, Hung Le, and Christian Wulff{-}Nilsen.
\newblock Greedy spanners are optimal in doubling metrics.
\newblock In {\em Proceedings of the Thirtieth Annual {ACM-SIAM} Symposium on
  Discrete Algorithms, {SODA} 2019, San Diego, California, USA, January 6-9,
  2019}, pages 2371--2379, 2019.

\bibitem[CC08]{CC-08}
Miroslav Chleb\'{\i}k and Janka Chleb\'{\i}kov\'{a}.
\newblock The {S}teiner tree problem on graphs: Inapproximability results.
\newblock {\em Theor. Comput. Sci.}, 406(3), October 2008.

\bibitem[CG06]{CG-06}
Richard Cole and Lee-Ad Gottlieb.
\newblock Searching dynamic point sets in spaces with bounded doubling
  dimension.
\newblock In {\em 38th annual ACM symposium on Theory of computing}, pages
  574--583, 2006.

\bibitem[CHJ16]{Chan-16}
T.{-}H.~Hubert Chan, Shuguang Hu, and Shaofeng~H.{-}C. Jiang.
\newblock A {PTAS} for the {S}teiner forest problem in doubling metrics.
\newblock In {\em FOCS 2016}, pages 810--819, 2016.

\bibitem[CJ16]{CJ-16}
T-H.~Hubert Chan and Shaofeng H.~C. Jiang.
\newblock Reducing curse of dimensionality: Improved ptas for {TSP} (with
  neighborhoods) in doubling metrics.
\newblock In {\em Proceedings of the Twenty-seventh Annual ACM-SIAM Symposium
  on Discrete Algorithms}, SODA '16, pages 754--765, 2016.

\bibitem[DHN93]{DHN93}
Gautam Das, Paul Heffernan, and Giri Narasimhan.
\newblock Optimally sparse spanners in 3-dimensional {E}uclidean space.
\newblock In {\em Proceedings of the Ninth Annual Symposium on Computational
  Geometry}, SCG '93, pages 53--62, New York, NY, USA, 1993. ACM.

\bibitem[EKM12]{EKM-12}
David Eisenstat, Philip Klein, and Claire Mathieu.
\newblock An efficient polynomial-time approximation scheme for {S}teiner
  forest in planar graphs.
\newblock In {\em Proceedings of the Twenty-third Annual ACM-SIAM Symposium on
  Discrete Algorithms}, SODA '12, pages 626--638, 2012.

\bibitem[FS16]{FS-16}
Arnold Filtser and Shay Solomon.
\newblock The greedy spanner is existentially optimal.
\newblock In {\em Proceedings of the 2016 ACM Symposium on Principles of
  Distributed Computing}, PODC '16, pages 9--17, 2016.

\bibitem[GGN06]{GGN-06}
Jie Gao, Leonidas~J. Guibas, and An~Nguyen.
\newblock Deformable spanners and applications.
\newblock {\em Comput. Geom. Theory Appl.}, 35(1), 2006.

\bibitem[Got15]{G-15}
Lee-Ad Gottlieb.
\newblock A light metric spanner.
\newblock In {\em Proceedings of the 2015 IEEE 56th Annual Symposium on
  Foundations of Computer Science (FOCS)}, FOCS '15, pages 759--772, 2015.

\bibitem[HM06]{HM-06}
Sariel {Har-Peled} and Manor Mendel.
\newblock Fast construction of nets in low-dimensional metrics and their
  applications.
\newblock {\em SIAM J. Comput.}, 35(5):1148--1184, 2006.

\bibitem[KL04]{KL-04}
Robert Krauthgamer and James~R. Lee.
\newblock Navigating nets: {S}imple algorithms for proximity search.
\newblock In {\em 15th Annual ACM-SIAM Symposium on Discrete Algorithms}, pages
  791--801, January 2004.

\bibitem[Mit99]{M-99}
Joseph S.~B. Mitchell.
\newblock Guillotine subdivisions approximate polygonal subdivisions: {A}
  simple polynomial-time approximation scheme for geometric {TSP}, $k$-{MST},
  and related problems.
\newblock {\em SIAM J. Comput.}, 28(4):1298--1309, 1999.

\bibitem[RS98]{RS98}
Satish~B. Rao and Warren~D. Smith.
\newblock Approximating geometrical graphs via ``spanners'' and ``banyans''.
\newblock In {\em 30th annual ACM symposium on Theory of computing}, pages
  540--550. ACM, 1998.

\bibitem[Smi10]{S-10}
Michiel Smid.
\newblock On some combinatorial problems in metric spaces of bounded doubling
  dimension.
\newblock Manuscript, available at
  \texttt{http://people.scs.carleton.ca/~michiel/research.html}, 2010.

\bibitem[Tal04]{T-04}
Kunal Talwar.
\newblock Bypassing the embedding: algorithms for low dimensional metrics.
\newblock In {\em 36th annual ACM symposium on Theory of computing}, pages
  281--290. ACM, 2004.

\end{thebibliography}

\end{document}